\documentclass{article}
\usepackage[noend]{algpseudocode}
\usepackage{authblk}

\errorcontextlines\maxdimen
\makeatletter
\newcommand*{\algrule}[1][\algorithmicindent]{\makebox[#1][l]{\hspace*{0.3em}\vrule height .75\baselineskip depth .25\baselineskip}}%

\newcount\ALG@printindent@tempcnta
\def\ALG@printindent{%
    \ifnum \theALG@nested>0
    \ifx\ALG@text\ALG@x@notext
    \addvspace{-1.35pt}
    \else
    \unskip
    \ALG@printindent@tempcnta=1
    \loop
    \algrule[\csname ALG@ind@\the\ALG@printindent@tempcnta\endcsname]%
    \advance \ALG@printindent@tempcnta 1
    \ifnum \ALG@printindent@tempcnta<\numexpr\theALG@nested+1\relax
    \repeat
    \fi
    \fi
}%

\usepackage{etoolbox}

\patchcmd{\ALG@doentity}{\noindent\hskip\ALG@tlm}{\ALG@printindent}{}{\errmessage{failed to patch}}

\makeatother

\usepackage{algorithm}
\usepackage{textcomp}
\usepackage{lineno,hyperref}
\usepackage{subfig}
\usepackage{graphicx}
\usepackage{amsmath}
\usepackage{amsthm}
\usepackage{amssymb}
\usepackage[utf8]{inputenc}
\usepackage{todonotes}
\usepackage{array}
\usepackage{booktabs,ragged2e}

\newcolumntype{P}[1]{>{\centering\arraybackslash}p{#1}}
\newcolumntype{M}[1]{>{\centering\arraybackslash}p{#1}}

\newtheorem{lemma}{Lemma}
\newtheorem{theorem}{Theorem}

\newcommand{\clique}[3]{\ensuremath{(#1,[#2,#3])}}

\newcommand{\ie}{{\em i.e.}}

\usepackage[]{todonotes}

\graphicspath{{img/}}

\begin{document}


\author[1]{Tiphaine Viard}
\author[2]{Cl\'{e}mence Magnien}
\author[2]{Matthieu Latapy}
\affil[1]{CEDRIC CNAM, F-75003 Paris, France}
\affil[2]{Sorbonne Universit\'{e}s, UPMC Univ Paris 06, UMR 7606, LIP6, F-75005, Paris, France}


\title{Enumerating maximal cliques\\ in link streams with durations}
\date{}

\maketitle

\begin{abstract}
Link streams model interactions over time, and a clique in a link stream is defined as a set of nodes and a time interval such that all pairs of nodes in this set interact permanently during this time interval. This notion was introduced recently in the case where interactions are instantaneous. 
We generalize it to the case of interactions with durations  and show that the instantaneous case actually is a particular case of the case with durations.
We propose an algorithm to detect maximal cliques that improves our previous one for instantaneous link streams, 
and performs better than the state of the art algorithms  in several cases of interest.
\end{abstract}



\section{Introduction}

A {\bf graph} is a pair of sets $G=(V,E)$ where $V$ is a set of nodes and $E\subseteq V\times V$ a set of links. If there is no distinction between $(u,v)$ and $(v,u)$ for $u$ and $v$ in $V$, then $G$ is undirected. If for all $(u,v)$ in $E$, $u \neq v$ then $G$ is simple. 
One generally considers undirected simple graphs. 
A clique of $G$ is a set of nodes $C\subseteq V$ such that all nodes of $C$ are linked to each other, 
\ie{} for all $\{u, v\}\subseteq C$, $(u,v)\in E$.
A clique $C$ is maximal if there is no other clique $C'$ such that $C\subset C'$. Enumerating maximal cliques of a graph is one of the most fundamental problems in computer science~\cite{Augustson1970, Bron1973, Tomita2006, Eppstein2011}.

An {\bf instantaneous link stream} is a triplet $L=(T,V,E)$ where $T$ is a time interval, $V$ a set of nodes and $E\subseteq T\times V\times V$ a set of instantaneous links. If there is no distinction between $(t,u,v)$ and $(t,v,u)$, then $L$ is undirected. If for all $(t,u,v)$ in $E$, $u \neq v$ then $L$ is simple. For any given duration $\Delta$, we introduced in~\cite{viard2016computing} the notion of {\bf $\Delta$-clique} in a simple undirected link stream $L=(T,V,E)$: it is a pair $C = (X, [x,y])$ with $X\subseteq V$ and $[x,y]\subseteq T$ such that for all $\{u, v\}\subseteq X$ and $I \subseteq [x,y]$ with $|I|=\Delta$ there is a $t$ in $I$ such that $(t,u,v) \in E$. In other words, there is an interaction between all pairs of nodes in $X$ at least once every $\Delta$ within $[x,y]$. A clique $C = (X, [x,y])$ is maximal if there is no other clique $C' = (X', [x',y'])$ such that $C'\neq C$, $X \subseteq X'$ and $[x,y] \subseteq [x',y']$. We proposed a first algorithm to enumerate all maximal $\Delta$-cliques of an instantaneous link stream~\cite{viard2016computing}, recently improved by adapting the Bron-Kerbosch algorithm~\cite{himmel2016enumerating,himmel2017adapting}.

A {\bf link stream with durations}, or simply link stream, is a triplet $L=(T,V,E)$ where $T$ is a time interval, $V$ a set of nodes and $E\subseteq T\times T\times V\times V$ a set of links such that for all links $(b,e,u,v)$ in $E$ we have $e\ge b$\,\footnote{Though
in theory we can consider that time is either discrete or continuous,
in all practical cases the datasets have an intrinsic time resolution
and time is therefore discrete in practice.}. 
We call $e-b$ the duration of the link. If there is no distinction between $(b,e,u,v)$ and $(b,e,v,u)$, then $L$ is undirected. If for all $(b,e,u,v)$ in $E$, $u \neq v$ and for all $(b,e,u,v)$ and $(b',e',u,v)$ in $E$, $[b,e]\cap[b',e']=\emptyset$ then $L$ is simple. In the remainder of this paper, unless explicitly specified, we will consider simple undirected link streams only.

A {\bf clique} in a link stream with durations is a pair $C = (X, [x,y])$ with $X\subseteq V$ and $[x,y]\subseteq T$ such that for all $\{u,v\}\subseteq X$ there is a link $(b,e,u,v)$ in $E$ such that $[x,y] \subseteq [b,e]$. In other words, all pairs of nodes in $X$ are continuously linked together from $x$ to $y$. A clique $C = (X, [x,y])$ is maximal if there is no other clique $C' = (X', [x', y'])$ such that $C'\neq C$, $X \subseteq X'$ and $[x,y] \subseteq [x',y']$. See Figure~\ref{fig:clique-example} for an illustration.

\begin{figure}[htbp]
	\begin{center}
		\includegraphics[width=0.7\linewidth]{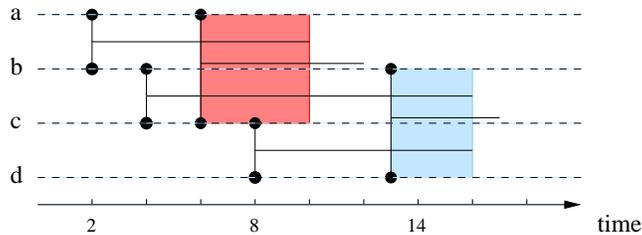}
		
        \caption{Maximal cliques in the link stream $L=(T,V,E)$, with $T=[0,20]$, $V=\{a,b,c,d\}$ and $E=\{(2,10,a,b), (4,16,b,c), (6,12,a,c), (8,16,c,d), (13,17,b,d)\}$. Two maximal cliques of $L$ are highlighted: in red, $(\{a,b,c\}, [6,10])$ is a maximal clique since $[6,10]$ is the largest interval over which nodes $a,b$, and $c$ all interact together and they don't interact with $d$ over this time interval. Similarly, $(\{b,c,d\}, [13,16])$ is a maximal clique, since $a$ does not interact with $b,c,d$ over $[13,16]$, and there is no larger interval such that $b,c,d$ all interact together. There are no other maximal cliques in $L$ involving $3$ nodes, but there are many other maximal cliques in $L$, such as $(\{c,d\}, [8,16])$ or $(\{a,c\}, [6,12])$ for instance.}
		\label{fig:clique-example}
	\end{center}
\end{figure}

Cliques in link streams with durations (and $\Delta$-cliques in instantaneous link streams)
bring valuable information in the study of different kinds of datasets;
for instance they indicate malicious computers coordinating their actions~\cite{Viard2018Discovering}.
Likewise, co-presence relations between animals is a key source of insight in ethology~\cite{MACCARRON2016119}, and cliques in the link streams with duration modeling such data may indicate
significant meetings. Many other fields may benefit from clique computations in link streams with durations in a similar way.

In this paper, we first extend our algorithm for maximal $\Delta$-cliques in instantaneous link streams~\cite{viard2016computing} to enumerate all maximal cliques in link streams with durations. 
The obtained algorithm is significantly simpler than the previous version,
and has a slightly lower complexity;
we show that it is possible to use it to enumerate maximal cliques in instantaneous link streams too, 
making it both more general and more efficient than our previous algorithm.
Experiments show that its running time is better than our previous algorithm,
as expected,
but also that it outperforms the more recent algorithm of Himmel {\em et al.}~\cite{himmel2016enumerating,himmel2017adapting}
in several cases of practical interest.

\section{Algorithm}

Like in~\cite{viard2016computing} our algorithm (Algorithm~\ref{algorithm}) relies on a set $S$ of previously computed cliques that we call candidates, and a set $M$ of already seen cliques. 
We initially populate both sets with  the trivial clique $(\{u,v\}, [b,b])$,
for all links $(b,e,u,v)\in E$ (Line~\ref{alg:init_state}) (finding cliques involving only one node is trivial and makes little sense, so we ignore them). 
Then, our algorithm iteratively picks and processes an element $(X, [x,y])$ from $S$ (Line~\ref{alg:get_clique}), until $S$ is empty (\emph{while} loop from Line~\ref{alg:begin_loop} to Line~\ref{alg:end_loop}). Processing $(X, [x,y])$ consists in searching for nodes $v\not\in X$ such that $(X\cup\{v\}, [x,y])$ is a clique (Lines~\ref{alg:add_node_begin} to~\ref{alg:add_node_end}), and for times $y'>y$ such that $(X,[x,y'])$ is a clique (Lines~\ref{alg:get_l} to \ref{alg:add_clique_time_e}).

For each node $v$ not in $X$, Line~\ref{alg:check_clique_node} checks that for all $u$ in $X$, there exists a link $(b,e,u,v)$ in the stream, with $[x,y]\subseteq [b,e]$. If $v$ satisfies this property, then $(X, [x,y])$ is not maximal (Line~\ref{alg:nodeismaxfalse}) and if $(X\cup\{v\}, [x,y])$ has not already been seen (Line~\ref{alg:ifXvinM}) then we add to both $S$ and $M$ (Line~\ref{alg:add_clique_node}).

The value of $l$ computed at Line~\ref{alg:get_l} is the largest time $l > y$ such that $(X, [x,l])$ is a clique. Line~\ref{alg:check_e_l} checks that this clique is different from $(X, [x,y])$, \ie{} $l \neq y$. In this case, $(X, [x,y])$ is not maximal (Line~\ref{alg:eismaxfalse}), and if the new clique $(X, [x,l])$ is new (Line~\ref{alg:ifeinM}) we add it to $S$ and $M$ (Line~\ref{alg:add_clique_time_e}).

If no node $v$ or time $l$ satisfies the conditions above, then the clique $(X, [x,y])$ is maximal and \emph{isMax} is true when we reach Line~\ref{alg:if_isMax}; we add the maximal clique to the output (Line~\ref{alg:add_c_r}).


\begin{algorithm}
\caption{Maximal cliques of a simple undirected link stream with durations}
\label{algorithm}
\renewcommand{\algorithmicrequire}{\textbf{input:} a simple undirected link stream with durations $L = (T,V,E)$}
\renewcommand{\algorithmicensure}{\textbf{output:} the set of all maximal cliques in $L$ involving at least two nodes}
\algorithmicrequire\\
\algorithmicensure

\begin{algorithmic}[1]
\State $S \gets \emptyset$, $R \gets \emptyset$, $M \gets \emptyset{}$
\State for $(b,e,u,v) \in E$: add $\clique{\{u,v\}}{b}{b}$ to $S$ and to $M$ \label{alg:init_state}
\While{$S \ne \emptyset$} \label{alg:begin_loop}
 \State take and remove $\clique{X}{x}{y}$ from $S$ \label{alg:get_clique}
 \State set isMax to True
 \For{$v$ in $V\setminus X$}\label{alg:add_node_begin}
  \If{$\clique{X\cup \{v\}}{x}{y}$ is a clique}\label{alg:check_clique_node}
   \State set isMax to False \label{alg:nodeismaxfalse}
   \If{$\clique{X\cup \{v\}}{x}{y}$ not in $M$} \label{alg:ifXvinM}
    \State add $\clique{X\cup \{v\}}{x}{y}$ to $S$ and $M$ \label{alg:add_clique_node} 
   \EndIf
  \EndIf
 \EndFor \label{alg:add_node_end}

 \State $l \gets \min(e: \exists (b,e,u,v)\in E \mbox{ such that } u,v\in X \mbox{ and } [x,y] \subseteq [b,e])$
	\label{alg:get_l}
 \If{$y \ne l$}\label{alg:check_e_l}
   \State set isMax to False \label{alg:eismaxfalse}
   \If{$\clique{X}{x}{l}$ not in $M$}\label{alg:ifeinM}
     \State add $\clique{X}{x}{l}$ to $S$ and $M$ \label{alg:add_clique_time_e}
   \EndIf
   \EndIf\label{alg:end_l}
 \If{isMax}\label{alg:test_max} \label{alg:if_isMax}
  \State add $\clique{X}{x}{y}$ to $R$\label{alg:add_c_r}
 \EndIf
\EndWhile\label{alg:end_loop}
\State \Return $R$
\end{algorithmic}
\end{algorithm}


\begin{theorem}[Correctness]
Given a simple undirected link stream with durations, Algorithm~\ref{algorithm} computes the set of all its maximal cliques involving at least two nodes.
\end{theorem}

\noindent
We first show that all the elements in the output of Algorithm~\ref{algorithm} are cliques, then that they are maximal, and finally that all maximal cliques are in this output.

\begin{lemma}
In Algorithm~\ref{algorithm}, all elements of $S$ are cliques of $L$.
\end{lemma}
\begin{proof}
$S$ initially contains cliques (Line~\ref{alg:init_state}) and Line~\ref{alg:add_clique_node} clearly preserves this property. The value $l$ computed at Line~\ref{alg:get_l} is the smallest value such that there exists a link of the form $(b,l,u,v) \in E$ for any two nodes $u,v \in X$. Since $[x,y] \subseteq [b,l]$ this means that $(X, [x,l])$ is a clique, and so the elements added at Line~\ref{alg:add_clique_time_e} are also cliques.
\end{proof}

\begin{lemma}
  \label{lem:max}
  All the elements of the set returned by Algorithm~\ref{algorithm} are maximal cliques of $L$.
\end{lemma}
\begin{proof}
A clique $(X, [x,y])$ from $S$ is added to $R$ only if isMax is True (Lines~\ref{alg:if_isMax} and~\ref{alg:add_c_r}).
If $(X, [x,y])$ is not maximal, 
then at least one of three conditions is true:
(1) there exists a node $v$ such that $\clique{X\cup \{v\}}{x}{y}$ is a clique, and then isMax is set to False at Lines~\ref{alg:check_clique_node} and~\ref{alg:nodeismaxfalse}, or (2) there exists a time $f < x$ such that $(X, [f,y])$ is a clique, or (3) there exists a time $l>y$ such that $(X, [x,l])$ is a clique. The second case cannot occur as all elements $(X, [x,y])$ of $S$ all involve (from its initialization and by recurrence) a link starting at $x$. If we are in the third case, then Line~\ref{alg:get_l} computes a value of $l$ satisfying $l>y$, which implies that the condition of Line~\ref{alg:check_e_l} is satisfied, and so Line~\ref{alg:eismaxfalse} sets isMax to False. Finally, if a clique of $S$ is not maximal it cannot be added to $R$.
%
\end{proof}

\medskip

\noindent
In order to prove the final result, we need the following lemma.


\begin{lemma}
  \label{lemma:uvc}
Let $L=(T,V,E)$ be a simple link stream, and let $C =(X, [x,y])$ be a maximal clique of $L$. Then there exists a link $(b,e,u,v)$ in $E$ such that $u$ and $v$ are in $X$ and $b=x$.
\end{lemma}
\begin{proof}
Assume this is false. Then, for all $u,v\in X$, there is a link $(b,e,u,v)$ such that $b < x$ and $e \ge y$. Then clearly $C$ is not maximal.
\end{proof}

\begin{lemma}
\label{lem:allmax}
The set returned by Algorithm~\ref{algorithm} contains all maximal cliques of $L$.
\end{lemma}
\begin{proof}
Let us consider a maximal clique $C = (X,[x,y])$ of $L$. If $C$ is in $S$ at some stage then it is easy to check that the algorithm adds it to $R$. We therefore show that $C$ is in $S$ at some stage.

Let us denote by $k$ the size of $X$, \ie{} $|X| = k$.
Let $u,v \in X$ such that there is a link $(x,e,u,v)$ in $E$; such nodes exist according to Lemma~\ref{lemma:uvc}.
Let $w_0=u$, $w_1 =v$, and let $(w_2, \ldots, w_{k-1})$ be an arbitrary ordering of all nodes in $X \setminus \{u,v\}$.
For all $i$, $1\le i \le k-1$, we consider the clique $C_i = (\{w_0,w_1, \ldots, w_i\}, [x,x])$.

We will show that, if $C_i$ is in $S$ at some stage, then the algorithm adds $C_{i+1}$ to $S$ when $C_i$ is the clique picked in $S$ at Line~\ref{alg:get_clique}.
Indeed, when this happens, Lines~\ref{alg:add_node_begin} to~\ref{alg:add_node_end} build $(X_i \cup \{w_{i+1}\}, [x,x]) = C_{i+1}$ from it and thus add $C_{i+1}$ to $S$ and $M$ at Line~\ref{alg:add_clique_node}. 

$C_1$ is added to $S$ by Line~\ref{alg:init_state}, and
if $C_i$ is in $S$ at some stage, then the algorithm adds $C_{i+1}$ to $S$. 
Therefore, $C_{k-1}$ is added to $S$ at some stage.

Consider now the iteration at which $C_{k-1}$ is picked from $S$ at Line~\ref{alg:get_clique}. Then, the value $l$ computed at Line~\ref{alg:get_l} is equal to $y$. Otherwise, either $l >y$ and then $C$ would not maximal, or $l<y$ and then $C$ would not be a clique. Therefore, the clique added to $S$ at Line~\ref{alg:add_clique_time_e} is $C$.
\end{proof}

\begin{theorem}[Complexity]
Given a link stream $L=(T,V,E)$ with $|V|=n$ and $|E|=m$, Algorithm~\ref{algorithm} enumerates all maximal cliques of $L$ in ${\cal O}(2^nn^3m^2+2^nn^2m^2\log m)$ time.
\end{theorem}
\begin{proof}
The complexity of Algorithm~\ref{algorithm} is dominated by the complexity of its main loop. The number of iterations of this loop is bounded by the number of elements added to $S$, which is bounded by the number of subsets of $V$ times the number of sub-intervals of $T$ of the form $[b,e]$, where $b$ is the beginning of a link in $E$ and $e$ is the end of a (different) link in $E$. Therefore, the number of iterations of the loop is in ${\cal O}(2^nm^2)$.

In the following,
we assume that sets are stored as binary search trees,
so that it is possible to search and add for an element in a set $T$ in ${\cal O}(\log |T|)$.
We also assume that all links are stored in a list, 
or binary tree, and are sorted by node pair then time,
so that it is possible to search for a link in ${\cal O}(\log m)$.

Now, let us consider a clique $(X, [x,y])$ picked from $S$ at Line~\ref{alg:get_clique}. Lines~\ref{alg:add_node_begin} to~\ref{alg:add_node_end} search for cliques of the form $(X\cup\{v\}, [x,y])$, $v\not\in X$, and Lines~\ref{alg:get_l} to~\ref{alg:add_clique_time_e} search for a clique $(X, [x,y'])$, $y'>y$. We analyze the two blocks separately.

For any $v\not\in X$, Line~\ref{alg:check_clique_node} checks if for all $u\in X$ there is a link $(b,e,u,v)$ such that $[x,y]\subseteq [b,e]$. This requires at most $|X|\log m$ steps, and so it is in ${\cal O}(n\log m)$. Line~\ref{alg:ifXvinM} searches for the newly found clique in $M$, which contains ${\cal O}(2^nm^2)$ elements. Since two cliques can be compared in ${\cal O}(n)$, this search is in ${\cal O}(n\log(2^nm^2)) = {\cal O}(n^2+n\log m)$. The algorithm repeats these operations for all $v\in V\setminus X$, so less than $n$ times. Finally, Lines~\ref{alg:add_node_begin} to~\ref{alg:add_node_end} have a cost in ${\cal O}(n(n\log m+n^2+n\log m) ) = {\cal O}(n^3 + n^2\log m)$. 

Computing $l$ at Line~\ref{alg:get_l} can be done in at most ${\cal O}(\frac{|X|\cdot|X-1|}{2}\log m)$ operations, and therefore is in ${\cal O}(n^2\log m)$. Lines~\ref{alg:ifeinM} and~\ref{alg:add_clique_time_e} can be done in ${\cal O}(n^2 + n\log m)$, as above. The complexity of Lines~\ref{alg:get_l} to~\ref{alg:add_clique_time_e} is then ${\cal O}(n^2\log m)$.

We conclude that each iteration of the main loop costs no more than ${\cal O}(n^3 + n^2\log m)$. We bound the overall complexity by multiplying this cost by the bound for the number of iterations of the loop. This leads to the ${\cal O}(2^nm^2(n^3 + n^2\log m)) = {\cal O}(2^nn^3m^2+2^nn^2m^2\log m)$ complexity.
\end{proof}

\section{Equivalence with $\Delta$-cliques}

Let us consider an instantaneous link stream $L=(T,V,E)$ with $T=[\alpha, \omega]$ and a value $\Delta \le \omega-\alpha$. We first define $T_\Delta$ as $[\alpha+\Delta,\omega]$, and, for all $u$ and $v$ in $V$, the set $T_{u,v} = \cup_{(t,u,v)\in E} [t,t+\Delta]$. We then define $E_\Delta$ as the set of all tuples $(b,e,u,v)$ such that $[b,e]$ is a maximal nonempty interval included in $T_\Delta \cap T_{u,v}$, for any $u$ and $v$ in $V$. We finally obtain the simple link stream with durations $L_\Delta = (T_\Delta,V,E_\Delta)$.

\begin{theorem}
$C=(X, [x,y])$ is a maximal $\Delta$-clique of $L$ if and only if $C_\Delta = (X, [x+\Delta,y])$ is a maximal clique of $L_\Delta$.
\end{theorem}
\begin{proof}

Assume $C$ is a $\Delta$-clique of $L$ and $C_\Delta$ is not a clique of $L_\Delta$. 
It means that there is a $t$ in $[x+\Delta,y]$ and $u$ and $v$ in $X$ such that 
$u$ and $v$ are not linked at time $t$ in $E_\Delta$.
This means that there is no link $(t',u,v)$ in $E$ with $t' \in [t-\Delta,t]$, which contradicts the assumption that $C$ is a $\Delta$-clique of $L$.

Conversely, assume $C_\Delta$ is a clique of $L_\Delta$ and $C$ is not a $\Delta$-clique of $L$. It means that there is an interval $I$ of duration $\Delta$ such that $I \subseteq [x,y]$ and for all $t$ in $I$, $(t,u,v) \not\in E$. 
If we denote by $t$ the largest element of $I$, it is at least equal to $x+\Delta$ and at most equal to $y$, 
and so there is no link containing $(t,u,v)$ in $E_\Delta$, which contradicts the assumption that $C_\Delta$ is a clique of $L_\Delta$.

Therefore, $C=(X, [x,y])$ is a $\Delta$-clique of $L$ if and only if $C_\Delta = (X, [x+\Delta,y])$ is a clique of $L_\Delta$. This is a bijection between the $\Delta$-cliques of $L$ and the cliques of $L_\Delta$ that preserves maximality.
\end{proof}

As a consequence, Algorithm~\ref{algorithm} may be used to compute the maximal $\Delta$-cliques in instantaneous link streams
(with complexity lower than
the initial algorithm published in~\cite{viard2016computing} 
by an improvement factor of $m/\log m$).

\section{Experiments}

We implemented Algorithm~\ref{algorithm}
and compared its running time with those of our previous implementation for computing $\Delta$-cliques~\cite{viard2016computing},
as well as the implementation provided by Himmel {\em et al.}
for their own algorithm~\cite{himmel2016enumerating,himmel2017adapting}.
All implementations are in Python\cite{cliques-lsCode2017}.

We used three datasets of different sizes coming from different contexts:
\begin{itemize}
\item {\sc Thiers-Highschool}~\cite{Fournet2014} is a
  trace of physical proximity between individuals, captured by
  sensors.  It was collected at a French high school in 2012 over a period of approximately 8 days.
  It contains 180 nodes and 45,047 links.
\item {\sc DNC-email} is the 2016 Democratic National Committee email leak~\cite{dnc-email},
  representing emails sent over a period of   approximately one year and four months\,\footnote{This duration covers the majority of emails;
    notice
    that a single email was sent approximately one year and a half before any other.}.
  It contains 1,866 nodes and  37,421 links.
\item {\sc Infectious} is a trace of physical proximity between visitors
  of an exhibition \cite{Isella2011What},
  collected over a period of approximately 80 days.
  It contains 10,972 nodes and  415,912 links.
\end{itemize}

\begin{figure}
  \centering
  \includegraphics[width=0.32\textwidth]{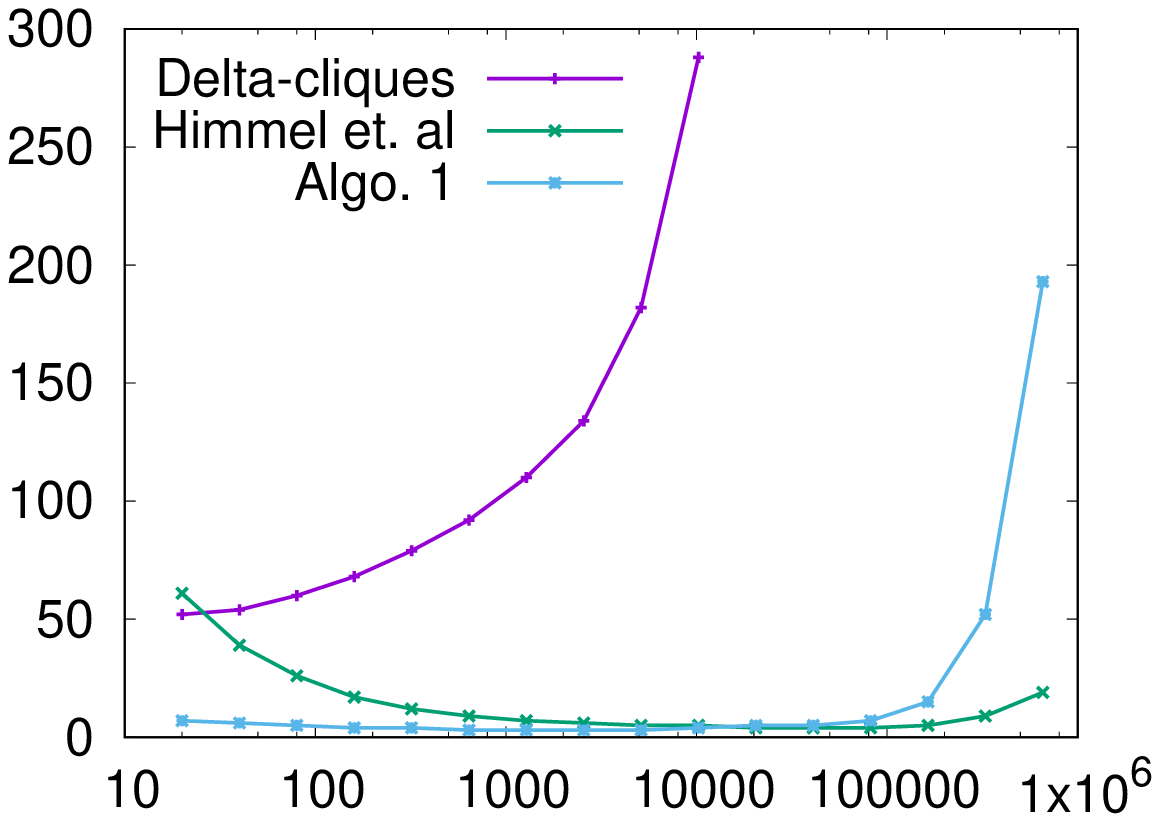}
  \includegraphics[width=0.32\textwidth]{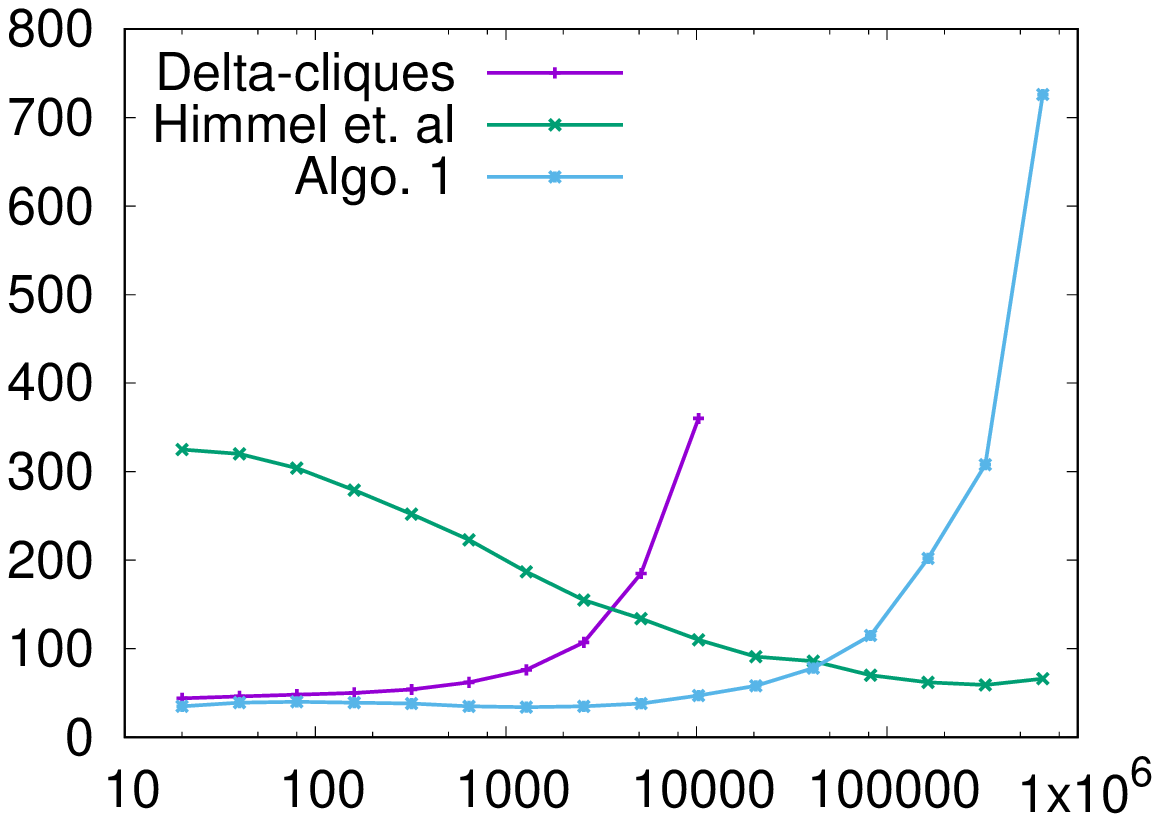}
  \includegraphics[width=0.32\textwidth]{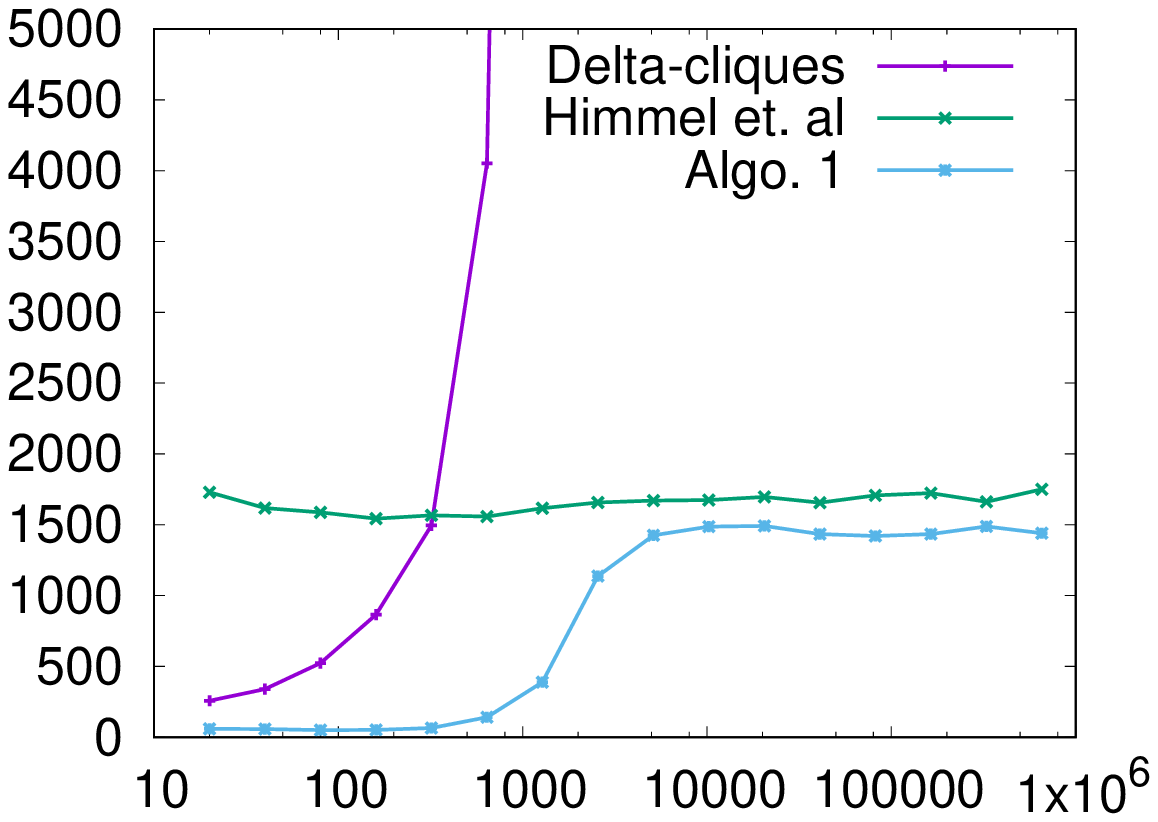}
  \caption{Running times (in seconds) for computing $\Delta$-cliques, for different values of $\Delta$
    (in seconds).
    Left: {\sc Thiers-Highschool}; middle: {\sc DNC-email}; right: {\sc Infectious}. }
  \label{fig:running_times_thiers}
\end{figure}

Results are presented in Figure~\ref{fig:running_times_thiers}.

We clearly observe two things.
First, our implementation significantly outperforms the code for $\Delta$-cliques 
for all values of $\Delta$.
Second, 
our implementation is the fastest for many relevant values of $\Delta$:
 in the cases of physical proximity our implementation is the fastest for all values of $\Delta$
 lower than 3 hours,
and in the case of emails,
it is the fastest for all values of $\Delta$ lower than 11 hours.
Notice that these values are of practical interest:  exchanging emails
at least once every hour within a group of people for a given time period, for instance, is significant.

In conclusion,
although practical evaluation of algorithms is difficult,
these experiments show that the most efficient solution depends on the target application,
and that our algorithm is appealing for small values of $\Delta$ at least in some cases of practical interest.




\bigskip
{\small\noindent\bf Acknowledgments.} This work is funded in part by the European Commission H2020 FETPROACT 2016-2017 program under grant 732942 (ODYCCEUS), by the ANR (French National Agency of Research) under grants ANR-15-CE38-0001 (AlgoDiv) and ANR-13-CORD-0017-01 (CODDDE), by the French program "PIA - Usages, services et contenus innovants" under grant O18062-44430 (REQUEST), and by the Ile-de-France program FUI21 under grant 16010629 (iTRAC).

\bibliographystyle{plain}

\bibliography{algo}

\end{document}